\def\be{\begin{equation}}
\def\ee{\end{equation}}
\def\ba{\begin{array}}
\def\ea{\end{array}}

\documentclass[prl,showpacs,twocolumn,amsmath]{revtex4}
\usepackage{amssymb}
\usepackage{pifont}
\usepackage{epsfig,subfigure,dsfont,amsthm,amsbsy,mathrsfs,amscd}
\def\qed{\leavevmode\unskip\penalty9999 \hbox{}\nobreak\hfill
     \quad\hbox{\leavevmode  \hbox to.77778em{%
               \hfil\vrule   \vbox to.675em%
               {\hrule width.6em\vfil\hrule}\vrule\hfil}}
     \par\vskip3pt}

\newtheorem{theorem}{Theorem}

\newtheorem{lemma}{Lemma}
\begin{document}
\title{Uniform Quantification of Correlations for Bipartite Systems}
\author{Tinggui Zhang$^{1}$, Hong Yang$^{2}$, Xianqing Li-Jost$^{1,3}$, Shao-Ming Fei$^{3,4}$}

\affiliation{ $^1$School of Mathematics and Statistics, Hainan
Normal University, Haikou, 571158, China\\
 $^{2}$ College of Physics
and Electronic Engineering, Hainan Normal University, Haikou 571158,
China\\
$^3$Max-Planck-Institute for Mathematics in the Sciences,
 Leipzig 04103, Germany\\
$^4$School of Mathematical Sciences, Capital Normal University,
Beijing 100048, P. R. China}

\begin{abstract}
Based on the relative entropy, we give a unified characterization of quantum correlations
for nonlocality, steerability, discord and entanglement for any bipartite quantum states. For
two-qubit states we show that the quantities obtained from quantifying
nonlocality, steerability, entanglement and discord
have strictly monotonic relationship. As for examples, the Bell
diagonal states are studied in detail.
\end{abstract}

 \maketitle

\section{Introduction}

As a key feature of quantum mechanics, quantum correlation has many different
forms such as entanglement, discord, steering, nonlocality
etc.. Nonlocality was first pointed out in 1935 by Einstein,
Podolsky and Rosen (EPR) \cite{aebp}, indicating that
the nonlocality must be an artefact of the
incompleteness of quantum mechanics. An interesting response to the EPR paradox was given by
Schr\"odinger \cite{esch}, who introduced another kind of correlation
in entangled states - quantum steering: Alice's ability to affect
Bob's state through her choice of measure basis. Later, it was shown that
the measurement outcomes from local measurements on entangled states
may be nonlocal, in the sense that they violate a Bell
inequality \cite{jsbe}. Let us briefly recall their definitions.

Quantum entanglement: Let $\rho_{AB}$ be a bipartite state with subsystems A and
B. $\rho_{AB}$ is said to be separable if it can be expressed as
$$
\rho_{AB}=\sum_{i=1}^n p_i \rho_{i}^A\otimes\rho_i^{B},
$$
where $0\le p_i\le 1$, $\sum\nolimits_i p_i=1$, $\rho_i^A$ and $\rho_i^B$
are states of the subsystems $A$ and $B$, respectively.
Otherwise it is called entangled. Basic aspects of entanglement including the characterization,
detection, distillation and quantification have been reviewed in \cite{rpmk}.

Bell nonlocality: Let $\mathcal{M}_a$ ($\mathcal{M}_b$) denote the set of observables
that Alice (Bob) performs measurement on the systems A (B).
Let $\lambda(A)$ ($\lambda(B)$) stand for the eigenvalues $\{a\}$ ($\{b\}$) of $A\in\mathcal{M}_a$
($B\in\mathcal{M}_b$),
and $P(a|A;W)$ ($P(b|B;W)$) the probability that Alice (Bob) gets the measurement outcome
$a$ ($b$) when she (he) measures the subsystem A (B) of the state $W$.
We say that a state $W$ is Bell local, if the following relation is satisfied for all $a \in
\lambda(a)$, $b \in \lambda(B)$, $A \in \mathcal{M}_a$ and $B \in
\mathcal{M}_b$,
\begin{eqnarray*}&&Tr[(\Pi_{a}^{A}\otimes \Pi_{b}^B)W]\\\nonumber &&\equiv P(a,b|A,B;W)=\sum_{\xi}p(a|A,\xi)p(b|B,\xi)p_{\xi},\end{eqnarray*}
where $\Pi_{a}^{A}$ ($\Pi_{b}^{B}$) is the projector satisfying
$A\,\Pi_{a}^{A}=a\Pi_{a}^A$ ($B\,\Pi_{b}^{B}=b\Pi_{b}^B$), $p(a|A,\xi)$ and $p(b|B,\xi)$
are some probability distributions involving the
local hidden variable (LHV) $\xi$. Recently, it has been
realized that one can significantly expand the notion of quantum
nonlocality by considering more complex causal structures going beyond
the usual LHV models \cite{chra,rccm}.

Quantum steering: If Alice performs the measurement $x$
and obtains the outcome $a$, then Bobs subnormalized reduced state
is given by $\rho_{a|x}=tr_A[(A_{a|x}\otimes I)\rho]$. And
$\sum_{a}\rho_{a|x}=\rho_B$ is independent of the measurement chosen
by Alice. The whole collection of ensembles $\{\rho_{a|x}\}_{a,x}$ is
a state assemblage. If there exists a local hidden state (LHS) model such that
\begin{eqnarray*}\label{defin0}\rho_{a|x}=\sum_{\lambda}p(\lambda)p(a|x,\lambda)\sigma_{\lambda},
\end{eqnarray*}
then Alice can not steer Bob's system. Here, $p(a|x,\lambda)$ are
some conditional probability distributions, and $\sigma_{\lambda}$
are a collection of subnormalized states that sum up to $\rho_B$
and satisfy $\sum_{\lambda}tr[p(\lambda)\sigma_{\lambda}]=1$
\cite{hmsj,mfpu,hzmh}. The steering problem is closely related to
the joint-measurement problem \cite{qmtv,urbc}.

Quantum discord: Two classically identical expressions
for the mutual information generally differ when the systems
involved are quantum. This difference was defined as the quantum
discord. It can be used as a measure of the quantumness of
correlations \cite{howh}. Later Ref. \cite{kmtp} gives the
mathematical definition of quantum discord of a state $\rho$,
$$
\mathcal{D}(\rho)=\min\limits_{\sigma\in \mathcal CC}S(\rho\parallel \sigma),
$$
where $\mathcal CC$ stands for the set of classically correlated
states of the form $\sum_{ij}p_{ij}|i\rangle\langle
i|\otimes|j\rangle\langle j|$, where $p_{ij}$ is a joint probability
distribution and $|i\rangle$ span the local orthonormal basis.

Although our understanding on different kinds of correlations has
advanced greatly recently, many fundamental questions remain open, e.g.,
(1) How to quantify these correlations? (2) What is the relationship
between them? To the first question there are already a lot of
literatures, see Ref.\cite{rpmk,oggt} for entanglement of quantum states,
Ref.\cite{lhvv,howh,kmtp} for quantifying quantum discord,
Ref. \cite{dcol,msee,jdcb,slsf} for the measure of nonlocality,
and Ref. \cite{psmn,ikar,Acrm} for quantifying steering.
There are many measures for different correlations. And the relationship
between them is rather complicated. It is even difficult to compare the measures for
a given correlation. Thus, one would ask if there is a unified
quantification for all the above correlations.

In this work, we give a unified
quantification for all these quantum correlations.
We first give the definitions of relative entropy steering and relative
entropy Non-locality. Then we study the relationship among them
for two-qubit states. At last, we discuss multipartite situations.

\section{Unified Quantification of Quantum Correlations}
The von Neumann relative entropy is defined as \cite{glin}
$S(\rho\|\sigma)\equiv tr\{\rho (\ln\rho - \ln\sigma)\}$
(in our text, we take $\log_2$ instead of $\ln$). In fact, Vedral et al.
\cite{vvmb,vvbm} first introduced the relative entropy of
entanglement, while the relative entropy of discord was first
proposed by Modi et al. \cite{kmtp}.
Let us list the above definitions \cite{yyxx}:
$$\mathcal{E}(\rho)=\min\limits_{\sigma\in
\mathcal{S}}S(\rho\parallel \sigma),$$
$$\mathcal{D}(\rho)=\min\limits_{\sigma\in \mathcal{CC}}S(\rho\parallel
\sigma),$$
where $\mathcal{E}(\rho)$ and $\mathcal{D}(\rho)$ are quantum entanglement
and quantum discord of state $\rho$, $\mathcal{S}$ and $\mathcal{CC}$ stand
for the sets of separable states and classically correlated states, respectively.

In the following, we choice von Neumann relative entropy to measure the quantum
steerability and nonlocality of quantum states. We first give a lemma.

\begin{lemma}The sets of all unsteerable states and LHV states are convex
sets, respectively. \end{lemma}

\begin{proof}
Any state admitting LHV models does not violate any Bell inequality.
Let $\rho_1$ and $\rho_2$ be two such LHV states. They satisfy all Bell inequalities like
$tr({\cal B}\rho_1) \leq c$ and $tr({\cal B}\rho_2) \leq c$, where ${\cal B}$ is any Bell
operators and $c$ some constant. Then one has $tr [{\cal B}(s \rho_1+(1-s) \rho_2)]=s [tr
({\cal B}\rho_1)]+(1-s) [tr({\cal B}\rho_2)] \leq  c$, where $s \in [0,1]$. This
proves that the LHV states constitute a convex set.

Suppose that states $\rho$ and $\widetilde{\rho}$ are
unsteerable. From (\ref{defin0}), there exists LHS model such that
$$
\rho_{a|x}=tr_A[(A_{a|x}\otimes I)\rho]=\sum_{\lambda}p(\lambda)p(a|x,\lambda)\sigma_{\lambda},
$$
where $p(a|x,\lambda) \geq 0$, and $\sigma_{\lambda}$ are a
collection of subnormalized states that sum up to $\rho_B$ which
satisfies $\sum_{\lambda}tr[p(\lambda)\sigma_{\lambda}]=1$ .
Analogously for $\widetilde{\rho}$,
$$
\widetilde{\rho}_{a|x}=tr_A[(A_{a|x}\otimes
I)\widetilde{\rho}]=\sum_{\mu}q(\mu)q(a|x,\mu)\sigma_{\mu}.
$$
Then
\begin{eqnarray*} \ \  &&[s\rho+(1-s)\widetilde{\rho}]_{a|x}=tr_A[(A_{a|x}\otimes
I)[s\rho+(1-s)\widetilde{\rho}]]\\[2mm]\nonumber &&=s[
tr_A[(A_{a|x}\otimes I)\rho]]+(1-s)tr_A[(A_{a|x}\otimes
I)\widetilde{\rho}]\\[1mm]\nonumber &&=s
\sum_{\lambda}p(\lambda)p(a|x,\lambda)\sigma_{\lambda}+ (1-s)
\sum_{\mu}q(\mu)q(a|x,\mu)\sigma_{\mu}\\\nonumber
&&=\sum_{\nu=1}^{m+n}p(\nu)p(a|x,\nu)\sigma_{\nu},
\end{eqnarray*}
where, $p(\nu)=s\,p(\lambda)$, $\sigma_{\nu}=\sigma_{\lambda}$ and $\nu=\lambda$
for $\nu=1,2,\cdots,m$; $p(\nu)=(1-s)\,p(\mu)$,
$\sigma_{\nu}=\sigma_{\mu}$ and $\nu=\mu$ for $\nu=m+1,m+2,\cdots,m+n$.
Then $p(a|x,\nu) \geq 0$, and
$\sigma_{\nu}$ are a collection of subnormalized states that sum up
to $s\rho_B+(1-s)\widetilde{\rho}_B$ which satisfies
$\sum_{\nu}tr[p(\nu)\sigma_{\nu}]
=s\sum_{\lambda}tr[p(\lambda)\sigma_{\lambda}]+(1-s)\sum_{\mu}tr[p(\mu)\sigma_{\mu}]=s+(1-s)=1$.
\end{proof}

From Lemma 1 we can define the following measure of steerability $\mathcal{S}(\rho)$ and
the measure of nonlocality $\mathcal{N}(\rho)$ for a quantum state $\rho$,
\begin{eqnarray}\label{defin1}
\mathcal{S}(\rho)=\min\limits_{\sigma\in \mathcal{U}}S(\rho\parallel
\sigma),
\end{eqnarray}
and
\begin{eqnarray}\label{defin2}
\mathcal{N}(\rho)=\min\limits_{\sigma\in \mathcal L}S(\rho\parallel
\sigma),
\end{eqnarray}
where $\mathcal{U}$ and $\mathcal{L}$ stand
for the sets of unsteerable states and the LHV states, respectively.
$S(\rho\parallel\sigma)=tr(\rho \log_2\rho-\rho \log_2\sigma)$.

The measure $\mathcal{S}(\rho)$ satisfy the following conditions ( analogously for $\mathcal{N}(\rho)$:
(1) $\mathcal{S}(\rho)\geq 0$, $\mathcal{S}(\rho)=0$ iff $\rho = \sigma $. (2) Local unitary
operations leave $\mathcal{S}$ invariant. (3) For any completely
positive trace preserving map $\Theta$,
$\mathcal{S}(\Theta\rho) \leq
\mathcal{S}(\rho)$ \cite{glin,vvmb,vvbm}. (4) $\mathcal{S}$ is convex, which can be proved in the following way. Let $\rho_1, \rho_2, \sigma_1, \sigma_2$ be four arbitrary states. From the convexity of the quantum relative entropy in both arguments \cite{mjdo}, we have
\begin{eqnarray*}
&&S(x\rho_1+(1-x)\rho_2||x\sigma_1+(1-x)\sigma_2) \\\nonumber &&\leq x S(\rho_1||\sigma_1)+(1-x)S(\rho_2||\sigma_2),
\end{eqnarray*}
where $x\in [0,1]$. By definition (\ref{defin1}) and Lemma 1, we have
\begin{eqnarray*}&&\mathcal{S}(x\rho_1+(1-x)\rho_2) \\\nonumber &&\leq S((x\rho_1+(1-x)\rho_2)||(x\sigma_1^*+(1-x)\sigma_2^*))\\\nonumber &&
\leq x S(\rho_1||\sigma_1^*)+(1-x) S(\rho_2 ||
\sigma_2^*)\\\nonumber &&=x S
(\rho_1)+(1-x)S(\rho_2),\end{eqnarray*} where $\sigma_i^*$ minimizes
$S(\rho_i||\sigma_i^*)$ over $\sigma\in \mathcal{U}$.

Since the steerability is weaker than Bell nonlocality and
stronger than nonseparability \cite{hmsj,Acrm}, one has
$\mathcal{CC} \subset \mathcal{S} \subset
\mathcal{U} \subset \mathcal{L}$. Therefore, we have the following relation:

\begin{theorem}\label{quancorr1}For bipartite quantum states,
the following relations hold:
\begin{equation}\label{tm}
\mathcal{D}(\rho) \geq \mathcal{E}(\rho) \geq \mathcal{S}(\rho) \geq \mathcal{N}(\rho).
\end{equation}
\end{theorem}
%\begin{figure}[h]
%\resizebox{7cm}{!}{\includegraphics{zhang01.eps}} \caption{the
%relationship of correlations}
%\end{figure}

As an example, let us consider two-qubit Bell-diagonal
states. Firstly, for two-qubit pure state
$|\phi\rangle=\alpha|00\rangle+\beta|11\rangle$, the
quantum discord \cite{sluo} and quantum entanglement \cite{vvmb,vvbm} are given by
$-|\alpha|^2\log_2|\alpha|^2-|\beta|^2\log_2|\beta|^2$.
It has been proven that every entangled pure state is steerable \cite{psmn}, and
all entangled pure states violate a
single Bell's inequality \cite{syqc}. Therefore, a separable pure
state is unsteerable and LHV. Now consider two-qubit Bell-diagonal states,
\begin{eqnarray}\label{defin3}
\rho_{AB}=\frac{1}{4}(I+\sum_{j=1}^3c_j\sigma_j^A\otimes
\sigma_j^{B})=\sum_{a,b=0}^1\lambda_{ab}|\beta_{ab}\rangle\langle
\beta_{ab}|,
\end{eqnarray}
where the $\sigma_j$'s are Pauli
operators \cite{mdcm}. The eigenstates are the four Bell states
$|\beta_{ab}\rangle \equiv (|0,b\rangle + (-1)^a|1,1\oplus
b\rangle)/\sqrt{2}$, with eigenvalues
$$\lambda_{ab}=\frac{1}{4}[1+(-1)^ac_1-(-1)^{a+b}c_2+(-1)^bc_3].$$

The quantum discord is given by $\mathcal{D}(\rho)=2-S(\rho_{AB})-C$, where
$C=\frac{1+c}{2}\log_2(1+c)+\frac{1-c}{2}\log_2(1-c)$ with $c= \max
|c_j|$ \cite{sluo,mdcm}.
For a Bell-diagonal state, when all $\lambda_{ab} \in
[0,\frac{1}{2}]$, $\mathcal{E}(\rho_{AB}) = 0$. When
$\lambda_{(00)} \geq \frac{1}{2}$ (analogously for other $\lambda_i \geq \frac{1}{2}$),
$\mathcal{E}(\rho_{AB})=\lambda_{(00)}\log_2(\lambda_{(00)})+(1-\lambda_{(00)})\log_2(1-\lambda_{(00)})+1$
\cite{vvmb}. Hence
\begin{eqnarray*}\mathcal{D}(\rho_{AB})-\mathcal{E}(\rho_{AB})&&=2-S(\rho_{AB})-C-E(\rho_{AB})\geq
0\end{eqnarray*}(See Fig. 1).
\begin{figure}[h]
\resizebox{6cm}{!}{\includegraphics{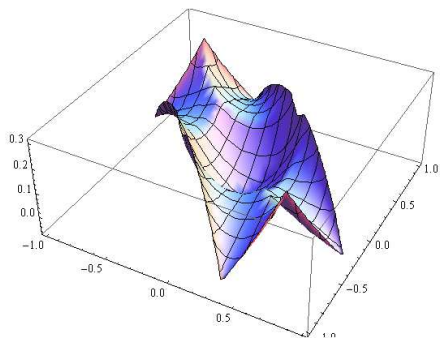}} \caption{The
quantity $\mathcal{D}(\rho)-\mathcal{E}(\rho)$ for $c_1,c_2\in
[-1,1]$, $c_3=0.4$.}
\end{figure}

In Ref. \cite{qqhz} it has been proved that any Bell-diagonal state with
correlation matrix $T$ is steerable by two projective measurements
iff $\lambda_1+\lambda_2 > 1$, where $\lambda_1$ and $\lambda_{2}$ are
the two largest eigenvalues of $TT^{t}$, where $t$ denotes transpose. For Bell diagonal state
(\ref{defin3}), the correlation matrix
$T=diag(c_1,c_2,c_3)$. Therefore, $\lambda_1+\lambda_2$ is the
summation of the two largest $c_i^2$. Without loss of generality, we
assume $c_1=c_2=\frac{\sqrt{2}}{2} > |c_3|$. In this case,
the Bell-diagonal state is unsteerable, that is, $\mathcal{S}(\rho)=0$.
But since $\lambda_{01} > \frac{1}{2}$, the entanglement is great than zero,
$\mathcal{E}(\rho)=\lambda_{01}\log_2\lambda_{01}-(1-\lambda_{01})\log_2(1-\lambda_{01})
> 0$. For example, for $|c_3|<0.4$, see Fig. 3.

Ref. \cite{qqhz} showed that a Bell diagonal state is steerable by
two projective measurements iff it violates the CHSH inequality.
This means that the steerability coincides with NLHV for Bell
diagnol states.

\begin{figure}[h]
\resizebox{6cm}{!}{\includegraphics{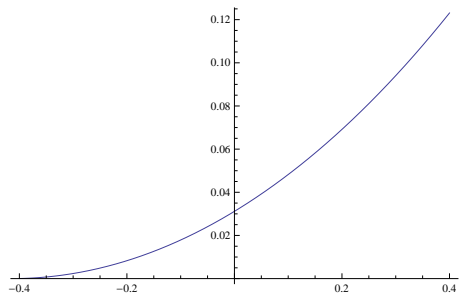}} \caption{$\mathcal{E}(\rho)$ for $|c_3| < 0.4$.}
\end{figure}

\section{conclusions and remarks}

We have studied some important correlations in bipartite systems.
Based on the relative entropy, we presented a unified characterization of discord, entanglement,
nonlocality and steerability. For two-qubit states we have showed that
the quantities obtained from quantifying nonlocality, steerability,
entanglement and discord have strictly monotonic
relationship. Detailed investigations have been given to Bell diagonal states.
The results can be generalized to the case of multipartite quantum systems.
In fact, from geometric point of view, any distance measure can chosen as a
candidate for quantifying quantum correlations.
One can choose any possible distance measure instead of $S(\rho\parallel
\sigma) $(relative entropy isn't a distance) in Eq.(\ref{defin1}) and
(\ref{defin2}) to quantify the Bell nonlocality and steerability. For
example, we can choose the Bures metric
$D(\rho||\sigma)=2-2\sqrt{F(\rho,\sigma)}$, where
$F(\rho||\sigma)\equiv
[tr\{\sqrt{\sigma}\rho\sqrt{\sigma}\}^{\frac{1}{2}}]^2$ is the
so-called fidelity \cite{rjos}.

Moreover, our approach also coincides with other methods in quantifying physical quantities besides
quantum correlations, like coherence \cite{tbmc},
$$
\mathcal{C}(\rho)=\min\limits_{\sigma\in
\mathcal{I}}S(\rho\parallel \sigma)=S(\rho_{diag})-S(\rho),
$$
where $\mathcal{I}$ and $\rho_{diag}$ stand
for incoherent states and the diagonal version of $\rho$.
Instead of (\ref{tm}), one has
$\mathcal{C}(\rho) \geq\mathcal{D}(\rho) \geq \mathcal{E}(\rho) \geq \mathcal{S}(\rho) \geq \mathcal{N}(\rho)$. For two-qubit Bell-diagonal states (\ref{defin3}),
the diagonal version of $\rho_{AB}$ is
$\rho_{(diag)}=\frac{1}{4}[I+c_3\sigma_3\otimes \sigma_3]$, with
eigenvalues $\Lambda_{(1,2)}=\frac{1}{4}[1+c_3]$ and
$\Lambda_{(3,4)}=\frac{1}{4}[1-c_3]$. Therefore, quantum coherence
$\mathcal{C}(\rho)=S(\rho_{(diag)})-S(\rho_{AB})=\sum_{a,b}\lambda_{ab}
\log_2(\lambda_{ab})-\sum_{i=1}^4\Lambda_i\log_2(\Lambda_i)$.
Therefore
\begin{eqnarray*}&&\mathcal{C}(\rho_{AB})-\mathcal{D}(\rho_{AB})\\\nonumber&&
=\frac{1+c}{2}\log_2(1+c)+\frac{1-c}{2}\log_2(1-c)\\\nonumber
&&-\frac{1+c_3}{2}\log_2(\frac{1+c_3}{4})
 -\frac{1-c_3}{2}\log_2(\frac{1-c_3}{4})-2\\\nonumber &&=
\frac{1+c}{2}\log_2(1+c)+\frac{1-c}{2}\log_2(1-c)\\\nonumber&&-\frac{1+c_3}{2}\log_2(1+c_3)-\frac{1-c_3}{2}\log_2(1-c_3)
\geq 0. \end{eqnarray*}
see Fig. 3.

\begin{figure}[h]
\resizebox{6cm}{!}{\includegraphics{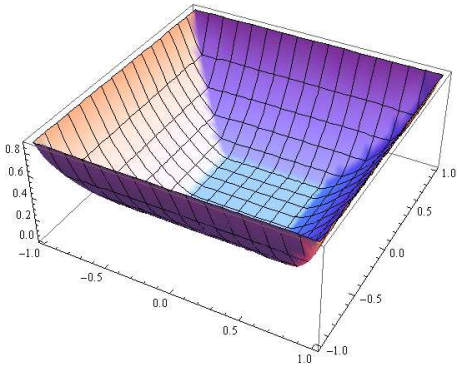}} \caption{
$\mathcal{C}(\rho_{AB})-\mathcal{D}(\rho_{AB})$ for $c_1,c_2\in [-1,1]$, $c_3=0.5$.}
\end{figure}

Due to the lack of the general expressions for unsteerable states and states admitting LHV models, it is difficult to calculate $\mathcal{S}(\rho)$ and $\mathcal{N}(\rho)$ for a given bipartite state $\rho$. It is neither an easy task to compute general $\mathcal{E}(\rho)$ (resp. $\mathcal{D}(\rho)$), although in these cases the general expressions of separable states (resp. zero-discord states) are explicitly known. We leave these problems for further investigations.

\bigskip
Acknowledgments:  This work is supported by the NSF of China under
Grant No. 11501153, 11675113; the NSF of Hainan Province
under Grant No.20161006.

\end{document}